\documentclass[11pt,a4paper]{llncs}

\usepackage{fullpage}
\usepackage{graphicx}
\usepackage{amsmath}
\usepackage{amssymb}
\usepackage{algorithm}
\usepackage{algorithmic}
\usepackage{hyperref}

\newcommand{\Order}{\mathrm{O}}
\newcommand{\OrderT}{\tilde{\mathrm{O}}}

\newcommand{\given}{\,|\,}

\renewcommand{\Pr}{\mathbb{P}}

\title{Optimal Lower Bounds for Matching and Vertex Cover in Dynamic Graph Streams\thanks{Jacques Dark was supported by an EMEA Microsoft Research studentship and European Research Council grant ERC-2014-CoG 647557.}}

\author{Jacques Dark\inst{1} and Christian Konrad\inst{2}}
\institute{Department of Computer Science, University of Warwick, Coventry, UK, \texttt{j.dark@warwick.ac.uk} \and 
Department of Computer Science, University of Bristol, Bristol, UK, \texttt{christian.konrad@bristol.ac.uk}}

\usepackage{tikz}
\usetikzlibrary{decorations.pathreplacing}

\usepackage{subcaption}
\usepackage{caption}

\newtheorem{problem}{Problem}

\begin{document}

 \maketitle
 
 \begin{abstract}
  In this paper, we give simple optimal lower bounds on the one-way two-party communication complexity of approximate \textsf{Maximum Matching}
 and \textsf{Minimum Vertex Cover} with deletions. In our model, Alice holds a set of edges 
 and sends a single message to Bob. Bob holds a set of edge deletions, which form a subset of Alice's edges,
 and needs to report a large matching or a small vertex cover in the graph spanned by the edges that are not deleted. 
 Our results imply optimal space lower bounds for insertion-deletion streaming algorithms for \textsf{Maximum Matching}
 and \textsf{Minimum Vertex Cover}. 
 
 Previously, Assadi et al. [SODA 2016] gave an optimal space lower bound for 
 insertion-deletion streaming algorithms for \textsf{Maximum Matching} via the simultaneous model of communication. 
 Our lower bound is simpler and stronger in several aspects:  
  The lower bound of Assadi et al. only holds for algorithms that (1) are able to process streams 
  that contain a triple exponential number of deletions in $n$, the number of vertices of the input graph; (2) 
  are able to process multi-graphs; and (3) never output edges that do not exist in the input graph when the randomized algorithm
  errs. In contrast, our lower bound even holds for algorithms that (1) rely on short ($\Order(n^2)$-length) input streams; 
  (2) are only able to process simple graphs; and (3) may output non-existing edges when the algorithm errs.
 \end{abstract}

 \section{Introduction}
Streaming algorithms for processing massive graphs have been studied for two decades \cite{hrrav98}. In the 
most traditional setting, the {\em insertion-only model}, an algorithm receives a sequence of the edges of the 
input graph in arbitrary order, and the objective is to solve a graph problem using as little space as possible. 
The insertion-only model has received significant attention, and many problems, such as matchings (e.g. \cite{kmm12,gkk12,kks14,kr16,kt17,k18,ps19,fhmrr20}), 
independent sets (e.g. \cite{hssw12,hhls16,cdk18,cdk19}), and subgraph counting (e.g. \cite{kmss12,cj17,bc17}), 
have since been studied in this model. See \cite{m14} for an excellent survey.

In 2012, Ahn et al. \cite{agm12} introduced the first techniques for addressing {\em insertion-deletion} graph streams,
where the input stream consists of a sequence of edge insertions and deletions. They showed that
many problems, such as \textsf{Connectivity} and \textsf{Bipartiteness}, can be solved using the same amount of space 
as in insertion-only streams up to poly-logarithmic factors. Various other works subsequently gave results of a similar 
flavor and presented insertion-deletion streaming algorithms with similar space complexity as their insertion-only
counterparts for problems including \textsf{Spectral Sparsification} \cite{kmmmnst20} and \textsf{$\Delta+1$-coloring} 
\cite{ack19}. Konrad \cite{k15} and 
Assadi et al. \cite{akly16} were the first to give a separation result between the insertion-only graph stream model and the 
insertion-deletion graph stream model: While it is known that a $2$-approximation to \textsf{Maximum Matching} can be computed using space 
$\Order(n \log n)$ in insertion-only streams, Konrad showed that space $\Omega(n^{\frac{3}{2} - 4 \epsilon})$ is required 
for an $n^{\epsilon}$-approximation in insertion-deletion streams, and Assadi et al. gave a lower
bound of $n^{2-3\epsilon-o(1)}$ for such an approximation. Assadi et al. also presented an $\OrderT(n^{2-3\epsilon})$ space 
algorithm that matches their lower bound up to lower order terms, which establishes that their lower bound is optimal 
(a different algorithm that matches this lower bound is given by Chitnis et al. \cite{ccehmmv16}).

Both Konrad and Assadi et al. exploit an elegant connection between insertion-deletion streaming algorithms
and linear sketches.
Ai et al. \cite{ahlw16}, building on the work of Yi et al. \cite{lnw14}, showed that insertion-deletion graph 
streaming algorithms can be characterized as algorithms that essentially solely rely on the computation of linear sketches of the input 
stream. A consequence of this result is that space lower bounds for insertion-deletion streaming algorithms
can also be proved in the {\em simultaneous model of communication}, since linear sketches can be implemented in this 
model. This provides an alternative to the more common approach of proving streaming lower bounds in the one-way model of 
communication. In particular, the lower bounds by Konrad and Assadi et al. are proved in the simultaneous model of communication.

From a technical perspective, this model has various attractive features, however, it comes with a major disadvantage:
The characterization of Ai et al. only holds for insertion-deletion streaming algorithms that (1) are able to process 
``very long'' input streams, i.e., input streams of triple exponential length in $n$, the number of vertices of the 
input graph, and (2) are able to process multi-graphs. In particular, this characterization does not hold for 
insertion-deletion streaming algorithms that rely on the assumption that input streams are short and the graph 
described by the input stream is always simple. Consequently, the lower bounds of Konrad and Assadi et al. 
do not hold for such algorithms.

\vspace{0.1cm}
\textbf{Our Results.}
In this work, we prove an optimal space lower bound for \textsf{Maximum Matching} in insertion-deletion streams 
via the one-way two-party model of communication. Our lower bound construction yields insertion-deletion streams of 
length $\Order(n^2)$ and does not involve multi-edges. Our lower bound therefore also holds for streaming algorithms that are 
designed for short input streams and simple graphs for which the characterization by Ai et al. does not hold.
Furthermore, the optimal lower bound by Assadi et al. \cite{akly16} only holds for streaming algorithms that %are able to process multigraphs,
%and 
never output non-existing edges when the (randomized) algorithms fail. We do not require this restriction.

Our lower bound method is simple and more widely applicable. Using the same method, we also give an optimal 
lower bound for \textsf{Minimum Vertex Cover}, showing that computing a $n^\epsilon$-approximation requires $\Omega(n^{2-2\epsilon})$
space. Assadi and Khanna mention in \cite{ak17} that the $n^{2 - 3\epsilon - o(1)}$ space lower bound 
for \textsf{Maximum Matching} given in \cite{akly16} also applies to \textsf{Minimum Vertex Cover}. Our lower bound
therefore improves on this result by a factor of $n^{{\epsilon} + o(1)}$. Furthermore, we show that our lower bound is 
optimal up to a factor of $\log n$: We give a very simple deterministic 
insertion-deletion streaming algorithm for \textsf{Minimum Vertex Cover} that uses space $\Order(n^{2-2\epsilon} \log n)$.

While the main application of our lower bounds in the one-way two-party communication model are lower bounds for insertion-deletion 
graph streaming algorithms, we believe that our lower bounds are of independent interest. Indeed, the one-way two-party communication 
complexity of \textsf{Maximum Matching} without deletions has been addressed in \cite{gkk12}, and our result can therefore also be understood
as a generalization of their model to incorporate deletions.

\vspace{0.1cm}
%\textbf{The Simultaneous and One-way Models of Communication}
\textbf{The Simultaneous Model of Communication.}
The lower bounds by Konrad \cite{k15} and Assadi et al. \cite{akly16} are proved in the simultaneous model of communication. 
In this model,
a typically large number of parties $k$ hold not necessarily disjoint subsets of the edges of the input graph.
Each party $P_i$ sends a message $M_i$ to a referee, who then outputs the result of the protocol. 
The connection between insertion-deletion streaming algorithms and linear sketches by Ai et al. \cite{ahlw16} then implies that 
a lower bound on the size of any message $M_i$ yields a lower bound on the space requirements of any 
insertion-deletion streaming algorithm. 

In the lower bound of Assadi et al. \cite{akly16} for \textsf{Maximum Matching}, each party $P_i$ holds the edges $E_i$
of a dense subgraph, which itself constitutes a {\em Ruzsa-Szemer\'{e}di graph}, i.e., a graph whose edge set can be partitioned 
into large disjoint induced matchings. All previous streaming lower bounds for approximate \textsf{Maximum Matching} rely on 
realizations of Ruzsa-Szemer\'{e}di graphs \cite{gkk12,k15,akly16}. 
Their construction is so that 
only a single induced matching of every party $P_i$ is useful for the construction of a global large matching. 
Due to symmetry of the construction, the parties are unable to identify the important induced matching and therefore 
need to send large messages that contain information about most of the induced matchings to the referee for them to 
be able to compute a large global matching. Interestingly, none of the parties hold edge deletions in their construction. 

\vspace{0.1cm}
\textbf{The One-way Model of Communication.}
In this paper, we give a lower bound in the one-way two-party model of communication. In this model, Alice holds
a set of edges $E$ of the input graph and sends a message $M$ to Bob. Bob holds a set of edge deletions $D \subseteq E$
and outputs a large matching in the graph spanned by the edges $E \setminus D$. A standard reduction shows that 
a lower bound on the size of message $M$ also constitutes a lower bound on the space requirements of an 
insertion-deletion streaming algorithm.
The two models are illustrated in Figure~\ref{fig:models}.

\begin{figure}[h!]
\begin{center}
 \begin{tikzpicture}
  \node (ref) at (0,0) {\textbf{Referee}};
  \node (out1) at (2,0) {result};
  \node (p1) at (-1.5, -1.5) {$\mathbf{P_1}$};
  \node (p2) at (-0.5, -1.5) {$\mathbf{P_2}$};
  \node (pdots) at (0.5, -1.5) {\dots};
  \node (p3) at (1.5, -1.5) {$\mathbf{P_k}$};
  
  \node(e1) at (-1.9, -2) {$E_1 \subseteq E$};
  \node(e2) at (-0.3, -2) {$E_2 \subseteq E$};
  \node(e3) at (1.5, -2) {$E_k \subseteq E$};
  
  \draw[->] (ref) -- (out1) {};
  \draw[->] (p1) -- (ref) node [midway, left,xshift=-2pt,yshift=-1pt] {$M_1$};
  \draw[->] (p2) -- (ref) node [midway, right,yshift=-1pt,xshift=-1pt] {$M_2$};
  \draw[->] (p3) -- (ref) node [midway, right,yshift=-1pt, xshift=4pt] {$M_k$};

  \node (alice) at (5, -0.5) {\textbf{Alice}};
  \node (bob) at (8, -0.5) {\textbf{Bob}};
  \node (out2) at (9.7,-0.5) {result};
  
  %\draw[->] ($ (alice) + (0.1,0) $) -- ($ (bob) + (-0.1,0) $) node [midway, above] {$M$};
  \draw[->] (alice) -- (bob) node [midway, above] {$M$};
  \draw[->] (bob) -- (out2) {};
  
  \node (ae1) at (5, -1) {$E$};
  \node (be2) at (8, -1) {$D \subseteq E$};
%  \node (be3) at (8, -1.5) {$(E_2)$};  
 \end{tikzpicture}
 \caption{The simultaneous (left) and the one-way two-party (right) models of communication. \label{fig:models}}
 \end{center}
\end{figure}
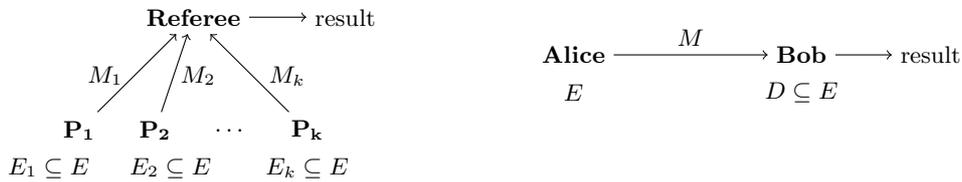

\vspace{0.1cm}
\textbf{Our Techniques.}
To prove our lower bound, we identify that an insertion-deletion streaming algorithm for \textsf{Maximum Matching}
or \textsf{Minimum Vertex Cover} can be used to obtain a one-way two-party communication protocol for a
two-dimensional variant of the well-known \textsf{Augmented Index} problem that we denote by \textsf{Augmented Bi-Index}, 
or \textsf{BInd} in short.
In an instance of \textsf{BInd}, Alice holds an $n$-by-$n$ binary matrix $A \in \{0, 1\}^{n \times n}$. 
Bob is given a position $(x,y) \in [n-k]^2$
and needs to output the bit $A_{x,y}$. Besides $(x, y)$, he also knows the $k$-by-$k$ submatrix of $A$ with upper left
corner at position $(x, y)$, however with the bit at position $(x,y)$ missing - we will denote this $k$-by-$k$ submatrix 
with $(x, y)$ missing by $A_{S(x, y)}$. We show that this problem has a one-way communication complexity of $\Omega((n-k)^2)$
by giving a reduction from the \textsf{Augmented Index} problem.

To obtain a lower bound for \textsf{Maximum Matching}, we show that Alice and Bob can construct a protocol for \textsf{BInd}
given an insertion-deletion streaming algorithm for \textsf{Maximum Matching}. In our reduction, we will consider instances
with $k = n - \Theta(n^{1-\epsilon})$, for some $\epsilon > 0$. Consider the following attempt: Suppose that the input matrix 
$A$ is a uniform random binary matrix and that $A_{x,y} = 1$ (we will get rid of these assumptions later). 
Alice and Bob interpret the matrix $A$ as the incidence matrix of a bipartite graph $G$. Bob interprets 
the ``1'' entries in the submatrix $A_{S(x,y)}$ outside the diagonal, i.e., all ``1'' entries except those in positions
$\{ (x+j, y+j) \ : \ 0 \le j < k \}$, as edge deletions $F$. 
The graph $G - F$ has a large matching: Since the diagonal of $A_{S(x,y)}$ is not deleted, and each entry in the diagonal is $1$
with probability $1/2$, we expect that half of all potential edges in the diagonal of $S(x,y)$ are contained in $G - F$ and thus form a matching of size 
$\Theta(k) = \Theta(n - n^{1-\epsilon})$. An $n^\epsilon$-approximation algorithm for 
\textsf{Maximum Matching} would therefore report $\Omega(n^{1-\epsilon})$ of these edges. Suppose that the
algorithm reported $\Omega(n^{1-\epsilon})$ uniform random edges from the diagonal in $A_{S(x, y)}$ (we will also get rid of this
assumption). Then, by repeating this scheme $\Theta(n^\epsilon)$ times in parallel, with large constant probability
the edge corresponding to $A_{x,y}$ is reported at least once, which allows us to solve \textsf{BInd}. 
This reduction yields an optimal $\Omega(n^{2-3\epsilon})$ space lower bound for insertion-deletion streaming algorithm for 
\textsf{Maximum Matching}, since $\Theta(n^{\epsilon})$ parallel executions are used to solve a problem that has a lower bound 
of $\Omega((n-k)^2) = \Omega(n^{2-2\epsilon})$.

In the description above, we assumed that (1) $A$ is a uniform random binary matrix; (2) $A_{x,y} = 1$; and (3) the algorithm outputs
uniform random positions from the diagonal of $A_{S(x, y)}$. To eliminate (1) and (2), Alice and Bob first sample
a uniform random binary matrix $X \in \{0, 1\}^{n \times n}$ from public randomness and consider the matrix obtained 
by computing the entry-wise XOR between $A$ and $X$, i.e., matrix $A \oplus X$, instead. Observe that $A \oplus X$ 
is a uniform random binary matrix (independently of $A$), and with probability $\frac{1}{2}$, property (2), i.e., $(A \oplus X)_{x,y} = 1$, holds. 
Regarding assumption (3), besides computing the XOR $A \oplus X$, Alice and Bob also sample two random 
permutations $\sigma_1, \sigma_2: [n] \rightarrow [n]$ from public randomness. Alice and Bob permute the rows 
and columns of $A \oplus X$ using $\sigma_1$ and $\sigma_2$, respectively.
Then, no matter which elements from the permuted relevant diagonal of $A \oplus X$ are reported by the algorithm, 
due to the random permutations, these elements could have originated from any other position in this diagonal. 
This in turn makes every element along the diagonal equally likely to be reported, including the position $(x,y)$ 
(in the unpermuted) matrix that we are interested in.

Our reduction for \textsf{Minimum Vertex Cover} is similar but simpler. We show that only a constant number of parallel executions
of an insertion-deletion streaming are required.

\vspace{0.1cm}
\textbf{Further Related Work.}
Hosseini et al. \cite{hly19} were able to improve on the ``triple exponential length'' requirement of the input streams
for a characterization of insertion-deletion streaming algorithms in terms of linear
sketches by Li et al. \cite{lnw14} and Ai et al. \cite{ahlw16}. They showed that in the case of XOR-streams 
and $0/1$-output functions, input streams of length $\Order(n^2)$ are enough.

Very recently, Kallaugher and Price \cite{kp20} showed that if either the stream length or the maximum value 
of the stream (e.g. the maximum multiplicity of an edge in a graph stream) are substantially restricted, 
then the characterization of turnstile streams as linear sketches cannot hold. For these situations they 
discuss problems where linear sketching is exponentially harder than turnstile streaming.

Besides the \textsf{Maximum Matching} problem, the only other separation result between the insertion-only and
the insertion-deletion graph stream models that we are aware of is a recent result by Konrad \cite{k19}, who showed that 
approximating large stars is significantly harder in insertion-deletion streams.

\vspace{0.1cm}
\textbf{Outline.}
We give a lower bound on the communication complexity of \textsf{Augmented Bi-Index} in Section~\ref{sec:aug-bi-index}.
Then, in Section~\ref{sec:matching}, we show that a one-way two-party communication protocol for \textsf{Maximum Matching} 
can be used to solve \textsf{Augmented Bi-Index}, which yield an optimal space lower bound for \textsf{Maximum Matching} 
in insertion-deletion streams. We conclude with a similar reduction for \textsf{Minimum Vertex Cover} in 
Section~\ref{sec:vertex-cover}, which also implies an optimal space lower bound for \textsf{Minimum Vertex Cover} in 
insertion-deletion streams.

\section{Augmented Bi-Index} \label{sec:aug-bi-index}
In this section, we define the one-way two-party communication problem \textsf{Augmented Bi-Index} and prove a
lower bound on its communication complexity. 

\begin{problem}[\textsf{Augmented Bi-Index}]
    In an instance of {\em \textsf{Augmented Bi-Index}} {\em $\textsf{BInd}^{n,k}_\delta$} we have two players denoted Alice and Bob.
    Alice holds a binary matrix $A \in \lbrace 0,1 \rbrace^{n \times n}$. Bob holds indices $x, y \in [n-k]$ and the incomplete\footnote{We use $A_S$ to refer to the collection of entries indexed by the set $S$, so $A_S = (A_{i,j})_{(i,j)\in S}$.} binary matrix $A_{S(x,y)}$ where
    \[
        S(x,y) = \{ (i,j) \in [n]^2 \given (x \leq i < x + k) \text{ and } (y \leq j < y + k) \} \setminus \{ (x,y) \} \ .
    \]    
    Alice sends a single message $M$ to Bob who must output $A_{x, y}$ with probability at least $1-\delta$.
\end{problem}

Our lower bound proof consists of a reduction from the well-known \textsf{Augmented Index} problem, which is known to 
have large communication complexity.

%\textsf{Augmented Index} is a two-party communication problem where Bob knows some subset of Alice's input, but is asked to return a particular part of the input he did not see.

\begin{problem}[\textsf{Augmented Index}]
    In an instance of {\em \textsf{Augmented Index}} {\em $\textsf{Ind}^{n}_\delta$} we have two players denoted Alice and Bob.
    Alice holds a binary vector $V \in \lbrace 0,1 \rbrace^n$. Bob holds an index $\ell \in [n]$ and the vector suffix $V_{>\ell} = (V_{\ell+1}, V_{\ell+2}, \cdots, V_n)$. 
    Alice sends a single message $M$ to Bob who must output $V_\ell$ with probability at least $1-\delta$.
\end{problem}

As a consequence of Lemma~13 in \cite{mnsw98}, we can see that this problem has linear communication complexity (see also Lemma~2 in \cite{bjkk04} for a more direct proof technique).

\begin{theorem}[e.g. \cite{mnsw98}]\label{thm:augind-lower}
    For $\delta < 1/3$, any randomised one-way communication protocol which solves $\textsf{Ind}^{n}_\delta$ must communicate $\Omega(n)$ bits.
\end{theorem}

%For our reductions we define a new variant \textsf{Augmented Bi-Index}, which introduces two-dimensional structure to both Alice's and Bob's inputs.

%Since only $(n-k)^2$ distinct bits of $A$ can be queried, there is an $(n-k)^2$ bit upper bound on the communication complexity and we can show this is tight by embedding instances of \textsf{Augmented Index} within $A$.

We are now ready to prove our lower bound for \textsf{Augmented Bi-Index}.
\begin{theorem}\label{thm:bind-lower}
    For $\delta < 1/3$, any randomised one-way communication protocol which solves $\textsf{BInd}^{n,k}_\delta$ must communicate $\Omega((n-k)^2)$ bits.
\end{theorem}
\begin{proof}
    Let $\mathcal{P}$ be a communication protocol for $\textsf{BInd}^{n,k}_\delta$ that uses messages of length at most 
    $S(n, k)$ bits. We will show how $\mathcal{P}$ can be used to solve $\textsf{Ind}^{(n-k)^2}_\delta$ with the same message size.
    
    Let $V, \ell$ be any instance of $\textsf{Ind}^{(n-k)^2}_\delta$. Alice builds the matrix $A \in \{ 0, 1 \}^{n \times n}$ by placing the bits of $V$ in lexicographical order in the top-left $(n-k)$-by-$(n-k)$ region:
    \[
        A_{i,j} =   \begin{cases}
                        V_{j+(n-k)(i-1)} &\text{for } i,j \in [n-k]\\
                        0 &\text{otherwise}
                    \end{cases} \ .
    \]    
    This packing is illustrated in Figure~\ref{fig:aug-index}(a).
    
    Alice runs protocol $\mathcal{P}$ on $A$ and sends the resulting message $M$ to Bob.    
    Now, Bob has the message $M$, the index $\ell \in [(n-k)^2]$ and the suffix $V_{>\ell}$. Let $x, y \in [n-k]$ be the unique pair of integers such that $\ell = y + (n-k)(x-1)$. Observe that $A_{x,y} = V_\ell$.
    
    For Bob to be able to complete protocol $\mathcal{P}$ he needs to provide $A_{S(x,y)}$. Because of the way we packed the entries of $V$ onto $A$, the overlap between $V$ and $A_{S(x,y)}$ is a subset of the entries of $V_{>\ell}$ (see Figure~\ref{fig:aug-index}(b) for an illustration). Therefore Bob can complete the protocol and determine $A_{x,y} = V_\ell$ with probability at least $1-\delta$.    
    By Theorem~\ref{thm:augind-lower}, it must be that $S(n, k) = \Omega((n-k)^2)$.
\end{proof}

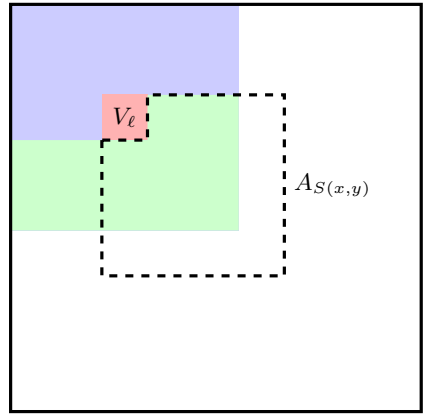
\begin{figure}
    \centering
    \hspace{1cm}
    \subcaptionbox{Example packing of the bits of $V$ into matrix $A$ with $n=9$ and $k=4$.}[0.4\textwidth]{
        \begin{tikzpicture}[
            cell/.style={rectangle, draw=black!35!white, minimum size =6mm, inner sep=0mm}
        ]
            \fill[blue!20!white] (3mm, -3mm) rectangle (33mm, -33mm);
            \foreach \y in {1,...,5}
            {
                \foreach \x in {1,...,5}
                {                
                    \pgfmathtruncatemacro{\i}{int(int((\y-1)*5 + \x))}
                    \node[cell] at (6*\x mm, -6*\y mm) {\small $V_{\i}$};
                }
                \foreach \x in {6,...,9} \node[cell] at (6*\x mm, -6*\y mm) {$0$};
            }
            \foreach \y in {6,...,9}
            {
                \foreach \x in {1,...,9} \node[cell] at (6*\x mm, -6*\y mm) {$0$};
            }
            \draw[black, very thick] (3mm, -3mm) rectangle (57mm, -57mm);
            \draw[black, very thick] (3mm, -3mm) rectangle (33mm, -33mm);
        \end{tikzpicture}
    }
    \hspace{1cm}
    \subcaptionbox{Bob can construct the area $A_{S(x,y)}$ given $V_{>\ell}$, which is part of his input.}[0.4\textwidth]{
        \begin{tikzpicture}[
            cell/.style={rectangle, minimum size =6mm, inner sep=0mm}
        ]
            \fill[blue!20!white] (3mm, -3mm) rectangle (33mm, -33mm);
            \node[cell, fill=red!30!white] at (18mm, -18mm) {\small $V_\ell$};
            \fill[green!20!white] (3mm, -21mm) -- (3mm, -33mm) -- (33mm, -33mm) -- (33mm, -15mm) -- (21mm, -15mm) -- (21mm, -21mm) -- (3mm, -21mm);
            \draw[very thick, dashed] (15mm, -21mm) -- (15mm, -39mm) -- (39mm, -39mm) -- node[right] {$A_{S(x,y)}$} (39mm, -15mm) -- (21mm, -15mm) -- (21mm, -21mm) -- (15mm, -21mm);
            \draw[black, very thick] (3mm, -3mm) rectangle (57mm, -57mm);
        \end{tikzpicture}
    }
    \caption{The construction of $A$ and $A_{S(x,y)}$ in Theorem~\ref{thm:bind-lower}.}
    \label{fig:aug-index}
\end{figure}

\section{Maximum Matching} \label{sec:matching}
Let $\mathbf{A}$ be a $C$-approximation insertion-deletion streaming algorithm for \textsf{Maximum Matching} that errs 
with probability at most $1/10$. We will now show that $\mathbf{A}$ can be used to solve $\textsf{BInd}^{n,k}_{\delta}$.

\subsection{Reduction}
Let $A \in \{0, 1\}^{n \times n}, x \in [n-k]$ and  $y \in [n-k]$ be an instance of $\textsf{BInd}^{n,k}_{\delta}$. 
Alice and Bob first sample a uniform random binary matrix $X \in \{0,1\}^{n \times n}$ 
and random permutations $\sigma_1, \sigma_2: [n] \rightarrow [n]$ from public randomness. Alice then computes matrix 
$A'$ which is obtained by first computing the entry-wise XOR of $A$ and $X$, denoted by $A \oplus X$, 
and then by permuting the rows and columns
of the resulting matrix by $\sigma_1$ and $\sigma_2$, respectively. Next, Alice interprets $A'$ as the incidence matrix
of a bipartite graph $G(A')$. Alice runs algorithm $\mathbf{A}$ on a random ordering of the edges of $G(A')$ and sends
the resulting memory state to Bob.

Next, Bob also computes the entry-wise XOR between the part of the matrix $A$ that he knows about, $A_{S(x,y)}$, and $X$, followed by applying
the permutations $\sigma_1$ and $\sigma_2$. In doing so, Bob knows the matrix entries of $A'$ at positions 
$(\sigma_1(i), \sigma_2(j))$
for every $(i,j) \in S(x,y)$. He can therefore compute the subset $E_S$ of the edges of $G(A')$ with

$$E_S = \{(\sigma_1(i),\sigma_2(j)) \in [n]^2 \ | \ (i,j) \in S(x,y) \mbox{ and } A'(\sigma_1(i),\sigma_2(j)) = 1  \} \ .$$

Furthermore, let $E_{diag} \subseteq E_S$ be the set of edges $(\sigma_1(i), \sigma_2(j))$ so that $(i,j)$ lies on the same
diagonal in $A$ as $(x,y)$, or, in other words, there exists an integer $1 \le q \le k-1$ such that $(x+q, y+q) = (i, j)$.
Then, let $E_{del} = E_S \setminus E_{diag}$.
Bob continues the execution of algorithm $\mathbf{A}$, as follows: for every edge $e \in E_{del}$, 
Bob introduces an edge deletion of $e$, in random order.

%For every edge $(i,j) \in E'$, if $\sigma_x^{-1}(i), \sigma_y^{-1}(j)$ does not lie on the same diagonal in $A$ as 
%$x,y$ then Bob simulates the edge deletion of edge $(i,j)$ in the stream. Let $M'$ be the matching returned by $\mathbf{A}$.
%Bob inserts deletions for every edge in $E''$.
Let $M'$ be the matching returned by $\mathbf{A}$.
From $M'$ Bob computes the matching $M$ as follows: If $|M' \le 0.99 \frac{k}{2C}|$ then Bob sets $M = \varnothing$.
%(either the 
%algorithm erred in this case or the very unlikely event happened that the maximum matching size in $G(A')$ is very small). 
Otherwise, Bob sets $M$ to be a uniform random subset of $M'$ of size exactly $0.99 \frac{k}{2C}$.

\textbf{Parallel Executions.} Alice and Bob execute the previous process $\ell = 100 \cdot C$ times in parallel. Let $M^i$, $X^i$, 
$\sigma_1^i$ and $\sigma_2^i$ be $M$, $X$, $\sigma_1$ and $\sigma_2$ that are used in run $i$, respectively.
%$M$ computed in run $i$, let $X^i$ be the matrix $X$ used in run $i$, and let $\sigma_1^i$ and $\sigma_2^i$ be the permutations used in run $i$. 
Let $Q_i$ be the indicator random variable that is $1$ iff $M^i$ contains the edge $(\sigma^i_1(x), \sigma^i_2(y))$. 
We also define $p = \sum_i Q_i$ to be the total number of times the edges $(\sigma^i_1(x), \sigma^i_2(y))$ are reported. 
Whenever the edge $(\sigma^i_1(x), \sigma^i_2(y))$ is reported, we interpret this to be a claim that $A_{x,y} = \neg X^i_{x,y}$. So depending on the value of $X^i_{x,y}$, this acts as a claim that $A_{x,y} = 0$ or $A_{x,y} = 1$. We define $p_0 = \sum_{i: Q_i = 1} X^i_{x, y}$ (which counts how often
$A_{x,y} = 0$ was claimed) and let $p_1 = p - p_0$ (the number of times $A_{x,y} = 1$ was claimed). Bob outputs $1$ as his estimator for $A_{x,y}$ if $p_1 \ge p_0$ 
and $0$ otherwise.

\subsection{Analysis}
Let $G$ be the bipartite graph with incidence matrix $A \oplus X$, and let 
$$F = \{ (i,j) \in S(x,y) \ | \ (A \oplus X)_{i,j} = 1 \mbox{ and } \nexists \ q \mbox{ s.t. } (i,j) = (x+q, y+q) \} \ .$$
Then the graph $G - F$ is isomorphic to the graph $G(A') - E_{del}$. In particular, 
$G(A') - E_{del}$ is obtained from $G - F$ by relabeling the vertex sets of the two bipartitions
using the permutations $\sigma_1$ and $\sigma_2$. 

We will first bound the maximum matching size in $G(A') - E_{del}$. To this end, we will bound the 
maximum matching size in $G - F$, which is easier to do:

\begin{lemma}\label{lem:matching-size}
 With probability $1-\frac{1}{k^{10}}$, the graph $G(A') - E_{del}$ is such that:
 $$0.99 \frac{k}{2} \le \mu(G(A') - E_{del}) \le 1.01 \frac{k}{2} + 2(n-k) \ ,$$
 where $\mu(G)$ denotes the \emph{matching number} of $G$, i.e., the size of a maximum matching.
\end{lemma}
\begin{proof}
 We will consider the graph $G - F$ instead, since it is isomorphic to $G(A') - E_{del}$ and has the same
 maximum matching size.
 
 First, observe that $G$ is a random bipartite graph where every edge is included with probability $\frac{1}{2}$.
 Let $U$ and $V$ denote the bipartitions in $G$, and consider the subsets $U' = [x, x+k)$ and $V' = [y, y+k)$.
 Observe that in the vertex induced subgraph $G[U' \cup V']$
 all edges are deleted in $F$ except those that connect 
 the vertices $x+i$ and $y+i$, for every $0 \le i \le k-1$. By a Chernoff bound, the number of edges and thus
 the maximum matching size in $G[U' \cup V']$ is bounded by:
 $$0.99 \cdot \frac{k}{2} \le \mu(G[U' \cup V']) \le 1.01 \cdot \frac{k}{2} \ , $$
 with probability $1 - \frac{1}{k^{10}}$.
 
 Observe that, with probability $1 - \frac{1}{k^{10}}$, the neighborhood $\Gamma(U')$ is such that 
 $$0.99 \cdot \frac{k}{2} \le |\Gamma(U')| \le 1.01 \cdot \frac{k}{2} + (n-k) \ .$$ 
 The set $U'$ can therefore be matched to at most $1.01 \cdot \frac{k}{2} + (n-k)$ vertices in $V$. 
 We thus obtain
 $$\mu(G - F) \le 1.01 \cdot \frac{k}{2} + 2(n-k) \ ,$$
 since we may also be able to match all $n-k$ vertices of $U \setminus U'$. 
\end{proof}

\begin{lemma}\label{lem:learn}
 Suppose that $M_i \neq \varnothing$. Then:
 \begin{eqnarray*}
\frac{0.99}{2C} - \frac{2(n - k)}{k} & \le & \Pr \left[ Q_i = 1 \right] \le \frac{0.99}{2C} \ .
%\Pr \left[ Q_i = 1 \ | \ \mbox{run $i$ errs} \right] & \le & \frac{1}{3C} \ .  
 \end{eqnarray*}
 
 %In particular, this holds regardless of the behaviour of $\mathbf{A}$ - whether it acts arbitrarily, or properly produces an approximate maximum matching. This means we can condition on $\mathbf{A}$ being correct or incorrect and still have the same bound.
\end{lemma}
\begin{proof}
 First, by construction of our reduction, since $M_i \neq \varnothing$ we have 
 $|M_i| = 0.99 \frac{k}{2C}$. 
 Let 
 $$U'_i = \sigma_1^i([x, x+k)) \mbox{ and } V'_i = \sigma_2^i([y, y+k)) \ .$$ 
 Let $\tilde{M}_i$ be the set of edges of $M_i$ connecting vertices in $U'_i$ to $V'_i$. Observe that there are
 $2(n-k)$ vertices in the graph outside the set $U'_i \cup V'_i$. We thus 
 have  
 $$|M_i| - 2(n-k) \le |\tilde{M}_i| \le |M_i|  \ .$$ 
 Next, since the permutations $\sigma^i_1, \sigma^i_2$ are chosen uniformly at random, any edge of $\tilde{M}_i$
 may have originated from any of the diagonal entries in $A_{S(x, y)}$. Hence, $\tilde{M}_i$ claims the bits
 of at least $|M_i| - 2(n-k)$ and at most $|M_i|$ uniform random positions in the diagonal of $A_{S(x, y)}$.
 Every entry in the diagonal of $A_{S(x,y)}$ is thus claimed with the same probability.
 Since the diagonal of $A_{S(x,y)}$ is of length $k$, this probability is at least
 $$ 
 \frac{|M_i| - 2(n-k)}{k} = \frac{0.99 \frac{k}{2C} - 2(n-k)}{k} = \frac{0.99}{2C} - \frac{2(n - k)}{k} \ ,
 $$
 and at most
 $$\frac{|M_i|}{k} = \frac{0.99 \frac{k}{2C}}{k} = \frac{0.99}{2C} \ .$$ 
\end{proof}

% be the set of edges with endpoints $i,j \in S(x,y)$. First, notice that the graph $G \setminus F$ is isomorphic to $G(A') \setminus E'$.
% We will bound the size of a maximum matching in $G - F$. 

% \begin{itemize}
%  \item With high probability, the largest matching in the graph is of size at least $0.99 \frac{k}{2} \le \mu(G(A'))$ and at 
%  most $\mu(G(A')) \le 1.01 \frac{k}{2} + (n-k)$.
%  \item Since the matchings $M_i$ are of size at least $0.99 \frac{k}{2C}$, they contain at least 
%  $0.99 \frac{k}{2C} - (n-k)$ edges from the important diagonal.
%  \item The probability that $A'_{\sigma_x(x), \sigma_y(y)} = 1$ is $1$ is $\frac{1}{2}$. The probability that $Q_i = 1$ is
%  hence $\frac{1}{2} \cdot \frac{0.99 \frac{k}{2C} - (n-k)}{\frac{1}{2} 1.01 k} = \frac{0.99 \frac{1}{2C}}{1.01} - o(1) \ge \frac{1}{3C} - o(1)$,
%  and all the $Q_i$ are independent. 
%  \item Suppose that $l = 120C$. Then, $\Exp p \ge 40$.
%  \item If the algorithm errs, then the probability that $Q_i = 1$ is at most $\frac{1}{3C}$. Hence, the number of times we get a
%  wrong prediction is at most in one third of the cases. For this to yield the majority of outcomes, the probability is very small.
% \end{itemize}

\begin{theorem}\label{thm:reduction} 
 Let $\mathbf{A}$ be a $n^{\epsilon}$-approximation insertion-deletion streaming algorithm 
 for \textsf{Maximum Matching} that errs with probability at most $1/10$ and uses space $s$. Then there exists a communication 
 protocol for $\textsf{BInd}^{n,n - \frac{1}{40} n^{1-\epsilon}}_{0.05}$ that communication $\Order(n^\epsilon \cdot s)$ bits.
\end{theorem}
\begin{proof}
Let $C = n^{\epsilon}$ and let $k = n - \frac{1}{40} n^{1-\epsilon}$.
 First, by Lemma~\ref{lem:matching-size}, with probability $1-\frac{1}{k^{10}}$, the graph 
 $G(A') - E_{del}$ contains a matching of size at least $0.99 k/ 2$. By a union bound, 
 the probability that this graph is of at least this size in each of the $\ell$ iterations  
 is at least $1 - \frac{\ell}{k^{10}}$. Suppose from now on that this event happens.

 Let $\ell_1$ be the number of times the algorithm $\mathbf{A}$ succeeds, and let $\ell_0$ be the number of times $\mathbf{A}$
 errs. Then, $\ell = \ell_0 + \ell_1$. Whenever $\mathbf{A}$ succeeds, since $\mathbf{A}$ is a $C$-approximation
 algorithm, the matching $M'_i$ is of size $0.99 \frac{k}{2C}$, which further implies that $M_i$ is of size exactly
 $0.99 \frac{k}{2C}$. Since the algorithm must return a correct matching, every time we have a claim (i.e. $Q_i = 1$), the claimed bit value must be correct. Thus, by Lemma~\ref{lem:learn}, we get a correct claim on $A_{x,y}$ with probability at least 
 \begin{eqnarray*}
  \frac{0.99}{2C} - \frac{2(n - k)}{k} & = & \frac{0.99}{2n^{\epsilon}} - \frac{2(\frac{1}{40} n^{1-\epsilon})}{n - \frac{1}{40} n^{1-\epsilon}} \ge  \frac{0.99}{2n^{\epsilon}} - \frac{\frac{1}{40} n^{1-\epsilon}}{n} = \frac{0.99}{2n^{\epsilon}} - \frac{1}{40n^{\epsilon}} \ge \frac{2}{5 n^\epsilon}   \ , \label{eqn:018}
 \end{eqnarray*}
where we used the inequality $\frac{2x}{y-x} \ge \frac{x}{y}$, which holds for every $y > x$. 
We thus expect to see the correct bit claimed at least $\ell_1 \cdot \frac{2}{5 n^\epsilon}$ times in total.
On the other hand, incorrect claims of the bit value can only occur when the algorithm errs. In the worst case, $\mathbf{A}$ 
will make as many false claims as possible - so we assume the algorithm never results in $M_i = \emptyset$ when it errs. 
Lemma~\ref{lem:learn} also allows us to bound the probability of an incorrect claim for this bad algorithm by $\frac{0.99}{2n^\epsilon}$. We thus expect to see the wrong bit value claimed at most $\ell_0 \cdot \frac{0.99}{2C} \le \frac{\ell_0}{2n^\epsilon}$ times.
 
 Recall that $\ell = 100 n^\epsilon$. Then, by standard concentration bounds, the probability that 
 $\ell_0 \ge 2 \cdot \frac{\ell}{10}$ is at most $\frac{1}{100}$ (recall that the error probability of $\mathbf{A}$ 
 is at most $\frac{1}{10}$). Suppose now that $\ell_0 \le \frac{1}{5} \ell$ holds, which also implies that $\ell_1 \ge \frac{4}{5} \ell$.
 We thus expect to learn the correct bit at least 
 $$\frac{4}{5} 100 n^{\epsilon} \cdot \frac{2}{5 n^\epsilon} = 32$$
 times, and using a Chernoff bound, it can be seen that the probability that we learn the correct bit less than 
 $21$ times is at most $0.02$. 
 Similarly, we expect to learn the incorrect bit at most 
 $$\frac{1}{5} 100 n^{\epsilon} \cdot \frac{1}{2n^\epsilon} = 10$$
 times, and by a Chernoff bound, it can be seen that the probability that we learn the incorrect bit at least $20$ times 
 is at most $ 0.01$. Our algorithm therefore succeeds if all these events happen. Taking a union bound
 over all failure probabilities that occurred in this proof, we see that our algorithm succeeds with probability
 $$1 - \frac{100n^{\epsilon}}{k^{10}} - 0.01 - 0.02 - 0.01 \ge 0.95 \ .$$
\end{proof}

Since by Theorem~\ref{thm:bind-lower}, $\textsf{BInd}^{n,n - \frac{1}{40} n^{1-\epsilon}}_{0.05}$ has randomized one-way 
communication complexity $\Omega(n^{2-2\epsilon})$, by Theorem~\ref{thm:reduction} we obtain our main result of this section:%lower bound result
%for \textsf{Maximum Matching} in insertion-deletion streams: %Since the hardness of $\textsf{BInd}^{n,n-\frac{1}{40}n^{1-\epsilon}}_{0.95}$ is $\Omega(n^{2-2\epsilon})$, we obtain:
\begin{corollary}
 Every insertion-deletion $n^\epsilon$-approximation streaming algorithm  
 for \textsf{Maximum Matching} that errs with probability at most $\frac{1}{10}$ requires space $\Omega(n^{2-3\epsilon})$.
\end{corollary}

\section{Minimum Vertex Cover} \label{sec:vertex-cover}
Let $\mathbf{B}$ be a $C$-approximation insertion-deletion streaming algorithm for \textsf{Minimum Vertex Cover} that succeeds with
probability $1 - 1/400$. Similar to the previous section, we will now show how $\mathbf{B}$ can be used to solve $\textsf{BInd}^{n,k}_{\delta}$.

\subsection{Reduction}
Let $A \in \{0, 1\}^{n \times n}, x \in [n-k]$ and  $y \in [n-k]$ be an instance of $\textsf{BInd}^{n,k}_{\delta}$. 
The reduction for \textsf{Minimum Vertex Cover} is very similar to the reduction for \textsf{Maximum Matching} presented 
in the previous section. Alice's behaviour is in fact identical: %For completeness, we give the entire reduction:

First, Alice and Bob sample a uniform random binary matrix $X \in \{0,1\}^{n \times n}$ and random permutations 
$\sigma_1, \sigma_2: [n] \rightarrow [n]$ from public randomness. 
Alice then computes matrix $A'$ which is obtained by first computing $A \oplus X$
and then  permuting the rows and then the columns of the resulting matrix by $\sigma_1$ and $\sigma_2$, respectively. 
Alice interprets $A'$ as the incidence matrix of a bipartite graph $G(A')$. Alice then runs algorithm $\mathbf{B}$ 
on a random ordering of the edges of $G(A')$ and sends the resulting memory state to Bob.

Next, Bob also computes the entry-wise XOR between the part of the matrix $A$ that he knows about and $X$, followed by applying
the permutations $\sigma_1$ and $\sigma_2$. In doing so, Bob knows the matrix entries of $A'$ at positions 
$(\sigma_1(i), \sigma_2(j))$
for every $(i,j) \in S(x,y)$. He can therefore compute the subset $E_S$ of the edges of $G(A')$ with

$$E_S = \{(\sigma_1(i),\sigma_2(j)) \in [n]^2 \ | \ (i,j) \in S(x,y) \mbox{ and } A'(\sigma_1(i),\sigma_2(j)) = 1  \} \ .$$

Next, Bob continues the execution of $\mathbf{B}$ and introduces deletions {\em for all edges in $E_S$} in random order. 
Observe that this 
step is different to the reduction for \textsf{Maximum Matching}. Let $I$ be the vertex cover produced by $\mathbf{B}$.
%If $|I| > C ( 2(n-k)+1)$ then Bob aborts and returns \texttt{fail}. 

\textbf{Parallel Executions.} Alice and Bob run the procedure above $40$ times in parallel. Denote by $I^i$, $X^i$, $E^i_S$,
$A'^i$, $\sigma_1^i$, and $\sigma_2^i$ the variables $I, X, E_S, A', \sigma_1$ and $\sigma_2$ used in iteration $i$.  
Furthermore, let $Q_i$ be 
the indicator variable that is $1$ iff $\{ \sigma^i_1(x), \sigma^i_2(y) \} \cap I_i \neq \varnothing$,
i.e., the potential edge $(\sigma^i_1(x), \sigma^i_2(y))$ is covered by the vertex cover. 

If there exists a run $j$ with $Q_j = 0$, then Bob predicts $A_{x,y} = X_{x,y}$ (if there are multiple such runs
then Bob breaks ties arbitrarily). Otherwise, Bob returns \texttt{fail} and the algorithm errs.

\subsection{Analysis}
%In the following, we assume that the algorithm $\mathbf{B}$ succeeds in each of the $40$ iterations.
%Since the failure probability of $\mathbf{B}$ is $\frac{1}{400}$, by the union bound, we encounter no failures 
%with probability at least $1 - \frac{1}{10}$. We account for the increase in the failure probability in the 
%proof of Theorem~\ref{thm:vc-reduction}. Furthermore, we will see that the graphs $G(A') - E_S$ contain a vertex cover
%of size at most $2(n-k) + 1$ (Lemma~\ref{lem:vc-size}). Since the approximation ratio of $\mathbf{B}$ is $C$, 
%we have that the vertex covers $I_j$, for $1 \le j \le 40$ are of sizes at most $C \cdot (2(n-k) + 1)$.

The first lemma applies to every parallel run $j$. For simplicity of notation, we will omit
the superscripts that indicate the parallel run in our random variables.

We first show an upper bound on the size of a minimum vertex cover in $G(A') - E_S$.

\begin{lemma}\label{lem:vc-size}
 The size of a minimum vertex cover in $G(A') - E_S$ is at most $2(n-k) + 1$.
\end{lemma}
\begin{proof}
Let $U, V$ be the bipartitions of the graph $G(A') - E_S$, let $U' = \{\sigma_1(a) \ : \ a \in [x, x+k) \}$ and let
$V' = \{\sigma_2(b) \ : \ b \in [y, y+k) \}$. Observe that $(G(A') - E_S)[U' \cup V']$ 
contains at most one edge: The potential edge between $\sigma_1(x)$ and $\sigma_2(y)$. A valid vertex cover 
of $G(A') - E_S$ is therefore $(U \setminus U') \cup (V \setminus V') + \sigma_1(x)$,
which is of size $2(n-k) + 1$.
\end{proof}

Next, we prove the key property of our reduction: We show that if $A'_{\sigma_1(x), \sigma_2(y)} = 0$ (or equivalently, 
$A_{x,y} \oplus X_{x,y} = 0$)
then neither $\sigma_1(x)$ nor $\sigma_2(y)$ is in the output vertex cover with large probability.

\begin{lemma} \label{lem:covered}
 Assume that algorithm $\mathbf{B}$ does not err in run $j$. Suppose that 
 $A'^j_{\sigma^j_1(x), \sigma^j_2(y)} = 0$. Then the probability that $Q_j = 1$ is at most 
 $$\frac{3 C \cdot (2(n-k) + 1)}{k} \ .$$
\end{lemma}
\begin{proof}
 Consider the set $D  = \{ (\sigma^j_1(x + i), \sigma^j_2(y+i)) \ | \ 0 \le i \le k-1 \}$, i.e., the positions of the diagonal
 of $S(x,y) \cup \{x, y\}$ permuted by $\sigma^j_1$ and $\sigma^j_2$. Then, since $A'^j$ is a uniform random matrix, with probability
 at least $1-\frac{1}{k^{10}}$, the ``permuted diagonal'' $A'^j_{D}$ contains at least $0.99 k/ 2$ entries with value $0$,
 or, in other words, graph $G(A'^j) - E^j_S$ contains at least $0.99 k/ 2$ non-edges in the positions of the permuted diagonal $D$.
 By Lemma~\ref{lem:vc-size}, the size of a minimum vertex cover in $G(A'^{j}) - E^{j}_S$ is at most $2(n-k) + 1$, and since
 $\mathbf{B}$ has an approximation factor of $C$, the  vertex cover $I_j$ is of size at most $C \cdot (2(n-k) + 1)$. Hence, 
 at most $C \cdot (2(n-k) + 1)$ non-edges in $D$ can be covered in $I_j$. However, since the permutations are random, the 
 probability that the non-edge $(\sigma^j_1(x), \sigma^j_2(y))$
 is covered, which is identical to the event $Q_j = 1$, is therefore at most 
 $$\frac{C \cdot (2(n-k) + 1)}{0.99 k/ 2} \le \frac{3 C \cdot (2(n-k) + 1)}{k} \ .$$
\end{proof}

\begin{theorem} \label{thm:vc-reduction}
 Let $\mathbf{B}$ be a $n^{\epsilon}$-approximation 
 insertion-deletion streaming algorithm for \textsf{Minimum Vertex Cover} that uses space $s$ and errs with probability 
 at most $1/400$. Then, there exists a communication protocol for $\textsf{BInd}^{n,n - \frac{1}{20}n^{1-\epsilon}}_{\frac{1}{3}}$ 
 that communicates $\Order(s)$ bits.
\end{theorem}
\begin{proof}
Let $k = n - \frac{1}{40} n^{1-\epsilon}$ and let $C= n^{\epsilon}$.
 Consider the reduction given in the previous subsection. First, observe that since $\mathbf{B}$ errs with probability
 at most $1/400$, by the union bound the probability that $\mathbf{B}$ errs at least once in the $40$ parallel executions of our
 reduction is at most $\frac{1}{10}$. We assume from now on that the algorithm never errs.
 
 Observe that the matrices $A'^j$ are random matrices. Hence, %we expect that in $20$ of the $40$ runs, 
 %we have that $A'^j_{\sigma^j_1(x), \sigma^j_2(y)} = 0$,
 %and 
 the probability that there exists at least one run $i$ with $A'^i_{\sigma^i_1(x), \sigma^i_2(y)} = 0$ is at least 
 $1 - (\frac{1}{2})^{40}$. Suppose that this event happens. Let run $i$ be so that $A'^i_{\sigma^i_1(x), \sigma^i_2(y)} = 0$.
 Then, by Lemma~\ref{lem:covered}, the probability that the non-edge $(\sigma^i_1(x), \sigma^i_2(y))$ is covered by $I_i$,
 or in other words, the probability that $Q_i = 1$, is at most 
 
 $$\frac{3 C \cdot (2(n-k) + 1)}{k} = \frac{3n^{\epsilon} \cdot (\frac{1}{20} n^{1-\epsilon} + 1) }{n - \frac{1}{40}n^{1-\epsilon}} = 
 \frac{\frac{3}{20}n + 3n^{\epsilon}}{n - \frac{1}{40}n^{1-\epsilon}} = \frac{3}{20} + o(1) \ . $$  
 
 Observe that whenever $Q_i = 0$, the algorithm outputs $X^i_{x,y}$ as a predictor for $A_{x,y}$. Since the algorithm $\mathbf{B}$ 
 does not err, we have $A_{x,y} \oplus X^i_{x,y} = 0$. This implies that $A_{x,y} = X^i_{x,y}$, which establishes correctness.
 
 Last, we need to bound the error probability of our algorithm. First, the probability that at least one of the $40$ runs 
 fails is at most $\frac{1}{10}$. Next, the probability that none of the runs are such that $A'^j_{\sigma^j_1(x), \sigma^j_2(y)} = 0$
 is at most $(\frac{1}{2})^{40}$. Furthermore, the probability that $Q_i = 1$ when $A'^i_{\sigma^i_1(x), \sigma^i_2(y)} = 0$ is at 
 most $\frac{3}{20} + o(1)$. Applying the union bound, we see that the overall error probability of our algorithm is at most
 $$\frac{1}{10}+ (\frac{1}{2})^{40} + \frac{3}{20} + o(1) \le \frac{1}{3} \ , $$
 for large enough $n$. 
\end{proof}
Since by Theorem~\ref{thm:bind-lower}, $\textsf{BInd}^{n,n - \frac{1}{40}n^{1-\epsilon}}_{\frac{1}{3}}$ has a 
communication complexity of $\Omega(n^{2-2\epsilon})$, we
obtain the following result:
\begin{corollary}
 Every insertion-deletion $n^\epsilon$-approximation streaming algorithm for \textsf{Minimum Vertex Cover} with error probability 
 at most $\frac{1}{400}$ requires space $\Omega(n^{2-2\epsilon})$.
\end{corollary}

\subsection{Insertion-deletion Streaming Algorithm for \textsf{Minimum Vertex Cover}}
We now sketch a simple deterministic $n^\epsilon$-approximation insertion-deletion streaming algorithm for 
\textsf{Minimum Vertex Cover} on general graphs that uses space $\Order(n^{2-2\epsilon} \log n)$. 
Let $G=(V, E)$ be the graph described by the input stream. The algorithm proceeds as follows: 
\begin{enumerate}
 \item Arbitrarily partition $V$ into subsets $V_1, V_2, \dots, V_{n^{1-\epsilon}}$,
 each of size $n^{\epsilon}$.
 \item Consider the multi-graph $G'$ obtained from $G$ by contracting the sets $V_i$ into vertices.
 \item \textbf{While processing the stream:} For each pair of vertices $V_i, V_j$ in $G'$ deterministically maintain the number 
 of edges connecting $V_i$ to $V_j$. 
 \item \textbf{Post-processing:} Compute a minimum vertex cover $I'$ in the multi-graph $G'$. 
 \item \textbf{Return} $I = \cup_{V_j \in I'} V_j$ as the vertex cover in $G$. 
\end{enumerate}

\noindent \textbf{Analysis:} Regarding space, the dominating space requirement is the maintenance of the number
of edges between every pair $V_i, V_j$. Since there are $n^{2-2\epsilon}$ such pairs, this requires space $\Order(n^{2-2\epsilon} \cdot \log n)$.

Concerning the approximation factor, let $I^*$ be a minimum vertex cover in $G$. 
Recall that $I'$ is an optimal cover in $G'$ and hence $|I'| \le |I^*|$ (edge contractions cannot increase the size of 
a minimum vertex cover). 
Since every set $V_j$ is of size $n^\epsilon$, 
the computed vertex cover $I$ is of size at most $|I'| \cdot n^{\epsilon} \le |I^*|n^{\epsilon}$, which proves the approximation factor.
By construction of the algorithm, every edge is covered.

\begin{theorem} \label{thm:vc-algorithm}
 There is a deterministic $n^\epsilon$-approximation insertion-deletion streaming algorithm for \textsf{Minimum Vertex Cover}
 that uses space $\Order(n^{2-2\epsilon} \log n)$.
\end{theorem}

\bibliography{dk20}

\end{document}